\documentclass[11pt,letterpaper]{article} 
\usepackage[parfill]{parskip} 
\usepackage{mathtools}
\usepackage{dsfont}
\usepackage[top=0.85in,left=1in,footskip=0.75in]{geometry}
\usepackage{algorithm} 
\usepackage[noend]{algpseudocode} 
\usepackage{array} 
\usepackage{float} 
\usepackage{amssymb} 
\usepackage{color} 
\usepackage{soul}
\usepackage[bordercolor=white,backgroundcolor=gray!30,linecolor=black, colorinlistoftodos]{todonotes}

\usepackage{subfigure} \providecommand{\tabularnewline}{\\}
\usepackage[english]{babel}
 \usepackage{amsthm}

\newtheorem{theorem}{Theorem}
\makeatother

\usepackage{changepage}

\usepackage[utf8]{inputenc}

\usepackage{textcomp,marvosym}

\usepackage{fixltx2e}

\usepackage{amsmath,amssymb}
\DeclareMathOperator*{\argmin}{arg\,min}

\usepackage{cite}

\usepackage{nameref,hyperref}

\usepackage[right]{lineno}

\usepackage{microtype} \DisableLigatures[f]{encoding = *, family = * }

\usepackage{rotating}

\usepackage[aboveskip=1pt,labelfont=bf,labelsep=period,justification= raggedright,singlelinecheck=off]{caption}

\makeatletter \renewcommand{\@biblabel}[1]{\quad#1.} \makeatother

\date{}

\usepackage{lastpage,fancyhdr,graphicx} 
\usepackage{epstopdf}
 \fancyheadoffset[L]{0in} \fancyfootoffset[L]{0in}  \usepackage{authblk}

\begin{document}
%
\title{
Improved prediction accuracy for disease risk mapping using Gaussian Process stacked generalisation} \author[1,*]{Samir Bhatt} \author[2]{Ewan Cameron} \author[4]{Seth R Flaxman} \author[2]{Daniel J Weiss} \author[3]{David L Smith}  \author[2]{Peter W Gething}

\affil[1]{Department of Infectious Disease Epidemiology, Imperial College London, London W2 1PG, UK} \affil[2]{Oxford Big Data Institute, Nuffield Department of Medicine, University of Oxford, Oxford OX3 7BN, UK} \affil[3]{Institute for Health Metrics and Evaluation, University of Washington, Seattle, Washington 98121, USA}\affil[4]{Department of Statistics, University of Oxford, 24-29 St Giles, Oxford OX1 3LB, UK}\affil[*]{Corresponding author: bhattsamir@gmail.com}

\date{} 

\maketitle

\begin{abstract} Maps of infectious disease---charting spatial variations in the force of infection, degree of endemicity, and the burden on human health---provide an essential evidence base to support planning towards global health targets. Contemporary disease mapping efforts have embraced statistical modelling approaches to properly acknowledge uncertainties in both the available measurements and their spatial interpolation. The most common such approach is that of Gaussian process regression, a mathematical framework comprised of two components: a mean function harnessing the predictive power of multiple independent variables, and a covariance function yielding spatio-temporal shrinkage against residual variation from the mean. Though many techniques have been developed to improve the flexibility and  fitting of the covariance function, models for the mean function have typically been restricted to simple linear terms. For infectious diseases, known to be driven by complex interactions between environmental and socio-economic factors, improved modelling of the mean function can greatly boost predictive power. Here we present an ensemble approach based on stacked generalisation that allows for multiple, non-linear algorithmic mean functions to be jointly embedded within the Gaussian process framework. We apply this method to mapping \emph{Plasmodium falciparum} prevalence data in Sub-Saharan Africa and show that the generalised ensemble approach markedly out-performs any individual method.
\end{abstract}

\pagebreak

\section*{Author Summary}

Infectious disease mapping provides a powerful synthesis of evidence in an effective, visually condensed form. With the advent of new web-based data sources and systematic data collection in the form of cross-sectional surveys and health facility reporting, there is high demand for accurate methods to predict spatial maps. The primary technique used in spatial mapping is known as Gaussian process regression (GPR). GPR is a flexible stochastic model that not only allows the modelling of disease driving factors such as the environment but can also capture unknown residual spatial correlations in the data.  We introduce a method that blends state-of-the-art machine learning methods with GPR to produce a model that substantially out-performs all other methods commonly used in disease mapping. The utility of this new approach also extends far beyond just mapping, and can be used for general machine learning applications across computational biology, including Bayesian optimisation and mechanistic modelling.

\section*{Introduction}

Infectious disease mapping with model-based geostatistics\cite{Diggle2007Model-basedGeostatistics} can provide a powerful synthesis of the available evidence base to assist surveillance systems and support progress towards global health targets, revealing the geographical bounds of disease occurrence and the spatial patterns of transmission intensity and clinical burden. A recent review found that out of 174 infectious diseases with a strong rational for mapping only 7 (4\%) have thus far been comprehensively mapped\cite{Hay2013GlobalDisease.b}, with a lack of accurate, population representative, geopositioned data having been the primary factor impeding progress. In recent years this has begun to change as increasing volumes of spatially referenced data are collected from both cross-sectional household surveys and web-based data sources (e.g. Health Map\cite{Freifeld2008HealthMap:Reports.}), bringing new opportunities for scaling up the global mapping of diseases. Alongside this surge in new data, novel statistical methods are needed that can generalise to new data accurately while remaining computationally tractable on large datasets.  In this paper we will introduce one such method designed with these aims in mind.

Owing to both a long history of published research in the field and a widespread appreciation amongst endemic countries for the value of cross-sectional household surveys as guides to intervention planning, malaria is an example of a disease that \textit{has} been comprehensively mapped. Over the past decade, volumes of publicly-available malaria prevalence data---defined as the proportion of parasite positive individuals in a sample---have reached sufficiency to allow for detailed spatio-temporal mapping \cite{Bhatt2015The2015.}. From a statistical perspective, the methodological mainstay of these malaria prevalence mapping efforts has been Gaussian process regression\cite{Rasmussen2006GaussianLearning,Bishop2006PatternLearning, Gething2011A2010,Hay2009A2007.}.  Gaussian processes are a flexible semi-parametric regression technique defined entirely through a mean function, $\mu(\cdot)$, and a covariance function, $k(\cdot,\cdot)$. The mean function models an underlying trend, such as the effect of environmental/socio-economic factors, while the covariance function applies a Bayesian shrinkage to residual variation from the mean such that points close to each other in space and time tend towards similar values. The resulting ability of Gaussian processes to strike a parsimonious balance in the weighting of explained and unexplained spatio-temporal variation has led to their near exclusive use in contemporary studies of the geography of malaria prevalence \cite{Hay2009A2007,Gething2011A2010,Diggle2007Model-basedGeostatistics, Gosoniu2012SpatiallyData.,Adigun2015MalariaData.,Bhatt2015The2015.}.

Outside of disease mapping, Gaussian processes have been used for numerous applications in machine learning, including regression\cite{Rasmussen2006GaussianLearning,Bishop2006PatternLearning, Diggle2007Model-basedGeostatistics}, classification\cite{Rasmussen2006GaussianLearning}, and optimisation\cite{Snoek2012PracticalAlgorithms}; their popularity leading to the development of efficient computational techniques and statistical parametrisations. A key challenge for the implementation of Gaussian process models arises in statistical learning (or inference) of any underlying parameters controlling the chosen mean and covariance functions. Learning is typically performed using Markov Chain Monte Carlo (MCMC) or by maximizing the marginal likelihood \cite{Rasmussen2006GaussianLearning}, both of which are made computationally demanding by the need to compute large matrix inverses returned by the covariance function. The complexity of this inverse operation is $\mathcal{O}(n^3)$ in computation and $\mathcal{O}(n^2)$ in storage in the naive case \cite{Rasmussen2006GaussianLearning}, which imposes practical limits on data sizes\cite{Rasmussen2006GaussianLearning}. MCMC techniques may be further confounded by mixing problems in the Markov chains. These challenges have necessitated the use of highly efficient MCMC methods, such as Hamiltonian MCMC\cite{homan2014no}, or posterior approximation approaches, such as the integrated nested Laplace approximation \cite{Rue2009ApproximateApproximations}, expectation propagation\cite{vanhatalo2010approximate,Minka2001ExpectationInference,Rasmussen2006GaussianLearning}, and variational inference\cite{hensman2013gaussian,Opper2009TheRevisited.}. Many of these methods  adopt finite dimensional representations of the covariance function yielding sparse precision matrices, either by specifying a fully independent training conditional (FITC) structure\cite{quinonero2005unifying} or by identifying a Gaussian Markov Random Field (GMRF) in approximation to the continuous process \cite{Lindgren2011AnApproachb}.

Alongside these improved methods for inference, recent research has focussed on model development to increase the flexibility and diversity of parametrisations for the covariance function, with new techniques utilising solutions to stochastic partial differential equations (allowing for easy extensions to non-stationary and anisotropic forms\cite{Lindgren2011AnApproachb}), the combination of kernels additively and multiplicatively\cite{duvenaud2013structure}, and various spectral representations\cite{wilson2013gaussian}.

One aspect of Gaussian processes that has remained largely neglected is the mean function which is often---and indeed with justification in some settings---simply set to zero and ignored. However, in the context of disease mapping, where the biological phenomena are driven by a complex interplay of environmental and socioeconomic factors, the mean plays a central role in improving the predictive performance of Gaussian process models. Furthermore, it has also been shown that using a well-defined mean function can allow for simpler covariance functions (and hence simpler, scalable inference techniques)  \cite{Fuglstad2015DoesFields}.

The steady growth of remotely-sensed data with incredible spatio-temporal richness \cite{Weiss2015Re-examiningApproach.} combined with well-developed biological models\cite{Weiss2014AirPrediction.} has meant that there is a rich suite of environmental and socio-economic covariates currently available. In previous malaria mapping efforts these covariates have been modelled as simple linear predictors \cite{Gething2011A2010,Hay2009A2007,Gosoniu2012SpatiallyData.} that fail to capture complex non-linearities and interactions, leading to a reduced overall predictive performance. Extensive covariate engineering can be performed by introducing large sets of non-linear and interacting transforms of the covariates, but this brute force combinatorial problem quickly becomes computationally inefficient\cite{Weiss2015Re-examiningApproach.,Bhatt2015The2015.}.

In the field of machine learning and data science there has been great success with algorithmic approaches that neglect the covariance and focus on learning from the covariates alone\cite{Hastie2009TheLearning,caruana2006empiricalb}. These include tree based algorithms such as boosting \cite{Friedman2001GreedyMachine} and random forests \cite{BreimanRandomForests}, generalized additive spline models \cite{HastieGam,Wood2006GeneralizedR}, multivariate adaptive regression splines \cite{Friedman1991MultivariateSplines}, and regularized regression models \cite{Zou2005RegularizationNet}. The success of these methods is grounded in their ability to manipulate the bias-variance tradeoff\cite{Geman1992NeuralDilemma}, capture interacting non-linear effects, and perform automatic covariate selection. The technical challenges of hierarchically embedding these algorithmic methods within the Gaussian process framework are forbidding, and many of the approximation methods that make Gaussian process models computationally tractable would struggle with their inclusion. Furthermore, it is unclear which of these approaches would best characterize the mean function when applied across different diseases and settings. In this paper we propose a simplified embedding method based on stacked generalisation \cite{Wolpert1992StackedGeneralization,BreimanStackedRegressions} that focuses on improving the mean function of a Gaussian process, thereby allowing for substantial improvements in the predictive accuracy beyond what has been achieved in the past.

\section*{Results}

\begin{figure} \centering \subfigure[]{\includegraphics[width=0.75\textwidth]{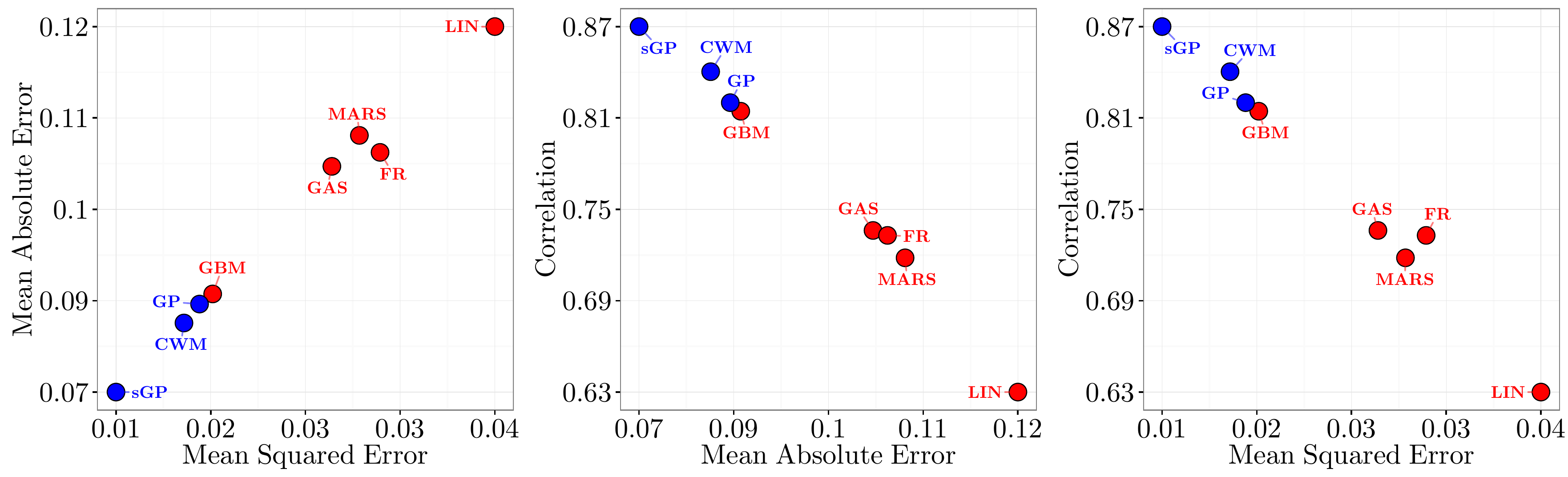}\label{fig:fig2a}} \subfigure[]{\includegraphics[width=0.75\textwidth]{Zone1.pdf}\label{fig:fig2b}} \subfigure[]{\includegraphics[width=0.75\textwidth]{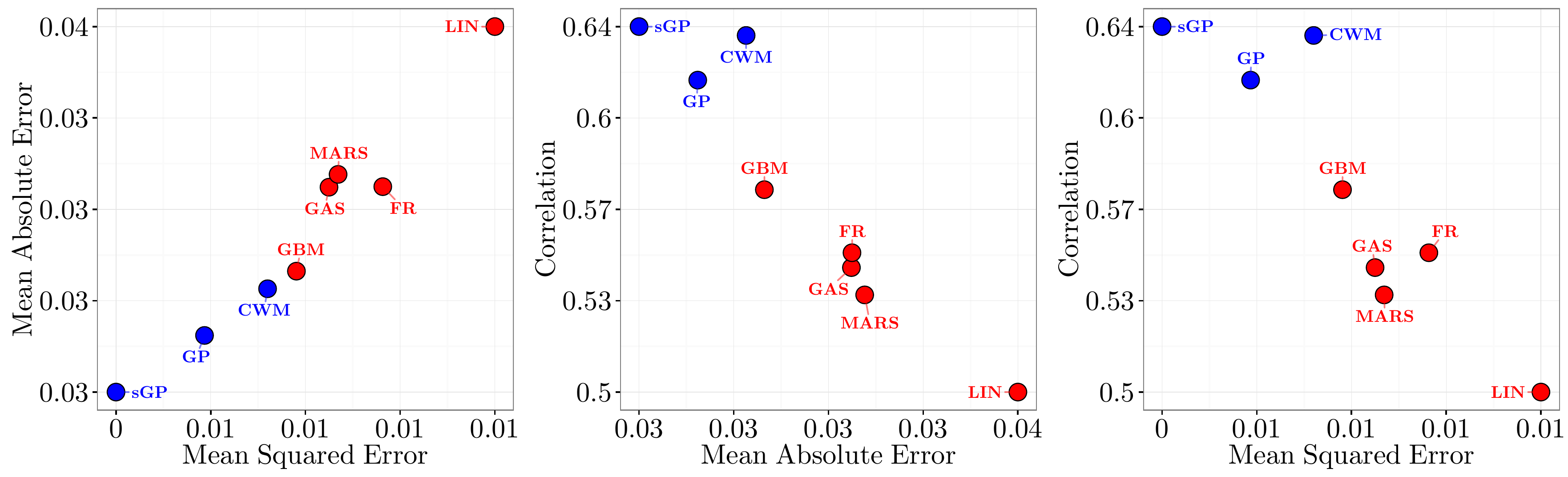}\label{fig:fig2c}} \subfigure[]{\includegraphics[width=0.75\textwidth]{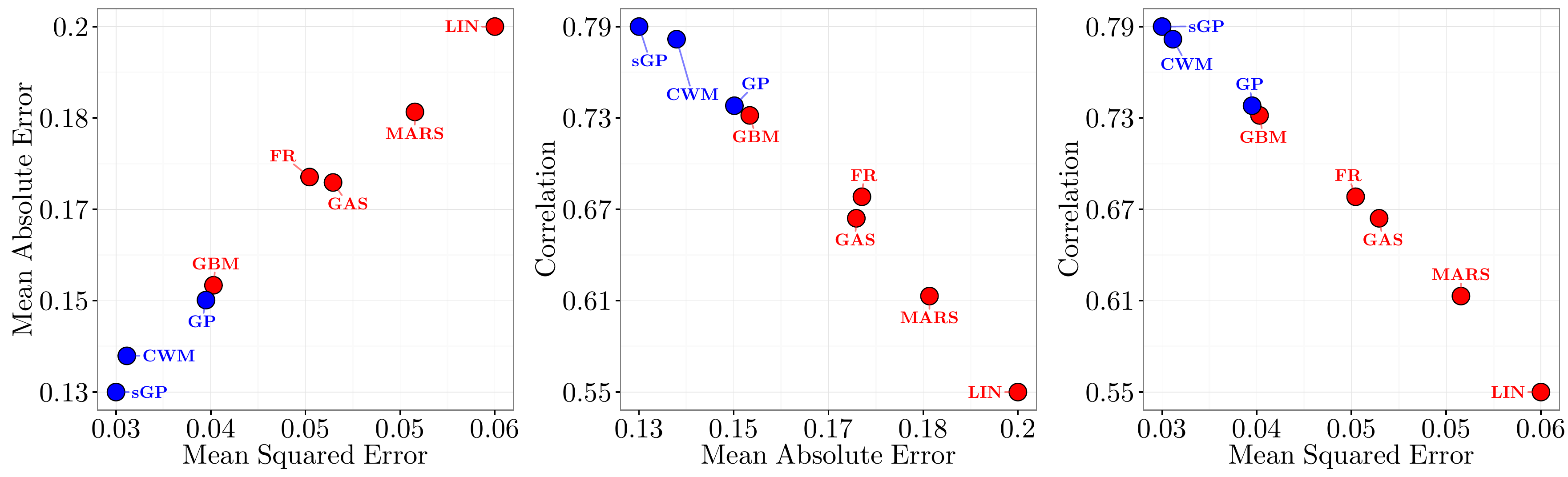}\label{fig:fig2d}} \caption{Comparisons of cross-validation MSE versus MAE versus correlation. Plot (a) is Eastern African, (b) Southern African, (c) North Eastern Africa and (d) Western Africa.  Legend abbreviations: (1) sGP - stacked Gaussian process, (2) CWM - stacked constrained weighted mean, (3) GP - standard Gaussian process, (4) GBM - Gradient boosted trees, (5) GAS - Generalised additive splines, (6) FR - Random forests, (7) MARS - Multivariate adaptive regression splines and (8) LIN - Elastic net regularised linear regression.}\label{fig:fig2} \end{figure}

The results of our analysis are summarised in figure \ref{fig:fig2} where pairwise comparisons of MSE versus MAE versus correlation are shown. We found a consistent ranking pattern in the generalisation performance of the methods across the Eastern, Southern and Western African regions (Figures \ref{fig:fig2a},\ref{fig:fig2b} and \ref{fig:fig2d}), with the stacked Gaussian process approach presented in this paper outperforming all other methods. The constrained weighted mean stacked approach was the next best method followed by the standard Gaussian process (with a linear mean) and Gradient boosted trees. Random forests, multivariate adaptive regression splines and generalised additive splines all had similar performance and the worst performing method was the elastic net regularised regression. For the North Eastern region (Figure \ref{fig:fig2c}), again the stacked Gaussian process approach was the best performing method, but the standard Gaussian process performed better than the constrained weighted mean stacked approach, though only in terms of MAE and MSE.

One average, across all regions, the stacked Gaussian process approach reduced the MAE and MSE by $9\% \; [1\% - 13\%]$ (values in square brackets are the min and max across all regions)  and $16\% \; [2\% - 24\%]$ respectively and increased the correlation by $3\% \; [1\% - 5\%]$ over the next best constrained weighted mean approach thereby empirically reinforcing the theoretical bounds derived in Appendix \ref{proof}. When compared the the widely used elastic net linear regression the relative performance increase of the Gaussian process stacked approach is stark, with reduced MAE and MSE by  $25\% \; [12\% - 33\%]$ and $25\% \; [19\% - 30\%]$ respectively and increase in correlation by $39\% \; [20\% - 50\%]$.

Compared to the standard Gaussian process previously used in malaria mapping the stacked Gaussian process approach reduced MAE and MSE by $10\% \; [3\% - 14\%]$  and $18\% \; [9\% - 26\%]$ respectively and increased the correlation by $6\% \; [3\% - 7\%]$.

Consistently across all regions the best non-stacked method was the standard Gaussian process with a linear mean function. Of the level zero generalisers gradient boosted trees were the best performing method, with performance close to that of  the standard Gaussian process. The standard Gaussian process only had a modest improvement over Gradient boosted trees with average reductions in MAE and MSE of $4\% \; [1\% - 8\%]$  and $7\% \; [1\% - 13\%]$ respectively and increases in correlation of $3\% \; [1\% - 7\%]$.

\section*{Discussion}

All the level zero generalisation methods used in this paper have been previously applied to a diverse set of machine learning problems and have track records of good generalisability \cite{Hastie2009TheLearning}. For example, in closely related ecological applications, these level zero methods have been shown to far surpass classical learning approaches\cite{Elith2008ATrees}. However, as introduced by Wolpert\cite{Wolpert1992StackedGeneralization}, rather than picking one level zero method, an ensemble via a second generaliser has the ability to improve prediction beyond that achievable by the constituent parts\cite{Krogh95neuralnetwork}. Indeed, in all previous applications \cite{BreimanStackedRegressions,Wolpert1992StackedGeneralization,vanderLaan2007SuperLearner,Sill2009Feature-WeightedStacking} ensembling by stacking has consistently produced the best predictive models across a wide range of regression and classification techniques. The most popular level one generaliser is the constrained weighted mean with convex combinations. The key attraction of this level one generaliser is the ease of implementation and theoretical properties \cite{Krogh95neuralnetwork,vanderLaan2007SuperLearner}. In this paper we show that, for disease mapping, stacking using Gaussian processes is more predictive and generalises better than both any single method in isolation, and the more common stacking approach using a constrained weighted mean.

The key benefit of stacking is summarised in equation \ref{eq:biasvariance} where the total error of an ensemble model can be reduced by using multiple, very different, but highly predictive models. However, stacking using a constrained weighted mean only ensures that the predictive power of the covariates are fully utilised, and does not exploit the predictive power that could be gained from characterising any residual covariance structure. The standard Gaussian process suffers from the inverse situation where the covariates are underexploited, and predictive power is instead gained from leveraging residual spatio-temporal covariance. In a standard Gaussian process the mean function is usually paramaterised through simple linear basis functions\cite{Rasmussen2004GaussianLearning} that are often unable to model the complex non linear interactions needed to correctly capture the true underlying mean. This inadequacy is best highlighted by the poor generalisation performance of the elastic net regularised regression method across all regions. The trade off between the variance explained by the covariates versus that explained by the covariance function will undoubtedly vary from setting to setting. For example in the Eastern, Southern and Western African regions, the constrained weighted mean stacking approach performs better than the standard Gaussian process, and the level 0 gradient boosted trees generaliser performs almost as well as the standard Gaussian process. For these regions, this shows a strong influence of the covariates on the underlying process. In contrast, for the North Eastern African region, the standard Gaussian process does better than both the constrained weighted mean approach (in terms of error not correlation) and all of the Level 0 generalisers, suggesting a weak influence of the covariates. However, for all zones, the stacked Gaussian process approach is consistently the best approach across all predictive metrics. By combining both the power of Gaussian processes to characterise a complex covariance structure, and multiple algorithmic approaches to fully exploit the covariates, the stacked Gaussian process approach combines the best of both worlds, and predicts well in all settings.

This paper introduces one way of stacking that is tailored for spatio-temporal data. However the same principles are applicable to purely spatial or purely temporal data, settings in which Gaussian process models excel. Additionally, there is no constraint on the types of level 0 generalisers than can be used; dynamical models of disease transmission e.g.~\cite{Smith2004StaticsMosquitoes.} \cite{Griffin2014EstimatesAfrica} can be fitted to data and used as the mean function within the stacked framework. Using dynamical models in this way can constrain the mean to include known biological mechanisms that can potentially improve generalisability, allow for forecast predictions and help restrict the model to only plausible functions when data is sparse. Finally multiple different stacking schemes can be designed (see Appendix \ref{alternative} for details), and relaxations on linear combinations can be implemented (e.g. \cite{Sill2009Feature-WeightedStacking}). 

Gaussian processes are increasingly being used for expensive optimisation problems \cite{OHagan1991BayesHermiteQuadrature} and Bayesian quadrature \cite{Hennig2015ProbabilisticComputations}. In current implementations both of these applications are limited to low dimensional problems typically with $n<10$. Future work will explore the potential for stacking to extend these approaches to high dimensional settings. The intuition is that the level 0 generalisers can accurately and automatically learn much of the latent structure in the data, including complex features like non-stationarity, which are a challenge for Gaussian processes. Learning this underlying structure through the mean can leave a much simpler residual structure\cite{Fuglstad2015DoesFields} to be modelled by the level one Gaussian process.

In this paper we have focused primarily on prediction, that is neglecting any causal inference and only searching for models with the lowest generalisation error. Determining causality from the complex relationships fitted through the stacked algorithmic approaches is difficult, but empirical methods such as partial dependence  \cite{Friedman2001GreedyMachine} or individual conditional expectation\cite{Goldstein2015PeekingExpectation} plots can be used to approximate the marginal relationships from the various covariates. Similar statistical techniques can also be used to determine covariate importance. 

Increasing volumes of data and computational capacity afford unprecedented opportunities to scale up infectious disease mapping for public health uses\cite{Pigott2015PrioritisingMapping.}.  Maps of diseases and socio-economic indicators are increasingly being used to inform policy \cite{Bhatt2013TheDengue.,Bhatt2015The2015.}, creating demand for methods to produce accurate estimates at high spatial resolutions. Many of these maps can subsequently be used in other models but, in the first instance, creating these maps requires continuous covariates, the bulk of which, come from remotely sensed sources. For many indicators, such as HIV or Tuberculosis, these remotely sensed covariates serve as proxies for complex phenomenon, and as such, the simple mean functions in standard Gaussian processes are insufficient to predict with accuracy and low generalisation error. The stacked Gaussian process approach introduced in this paper provides an intuitive, easy to implement method that predicts accurately through exploiting information in both the covariates and covariance structure.

\section*{Methods}
\subsection*{Gaussian process regression}

We define our response, $\mathbf{y}_{s,t} = \{y_{(s,t)[1]},...,y_{(s,t)[n]}\}$, as a vector of $n$ empirical logit transformed malaria prevalence surveys at location-time pairs, $(s,t)[i]$, with $\mathbf{X}_{s,t}=\{(\mathbf{x}_{1:m})[1],\ldots,(\mathbf{x}_{1:m})[n]\}$ denoting a corresponding $n\times m$ design matrix of $m$ covariates (see Section \ref{covs}). The likelihood of the observed response is $\mathbb{P}(\mathbf{y}_{s,t}|\mathbf{f}_{s,t},\mathbf{X}_{s,t},\theta)$, which we will write simply as $\mathbb{P}(y|f(s,t),\theta)$, suppressing the spatio-temporal indices for ease of notation.  Naturally, $f(s,t)$ $[=\mathbf{f}_{s,t}]$ is the realisation of a Gaussian process with mean function, $\mu_\theta(\cdot)$, and covariance function, $k_\theta(\cdot,\cdot)$, controlled by elements of a low-dimensional vector of hyperparameters, $\theta$.  Formally, the Gaussian process is defined as an $(s,t)$-indexed stochastic process for which the joint distribution over any finite collection of points, $(s,t)[i]$, is multivariate Gaussian with mean vector, $\mu_i = m((s,t)[i]|\theta)$, and covariance matrix, $\Sigma_{i,j}=k((s,t)[i],(s,t)[j]|\theta)$. The Bayesian hierarchy is completed by defining a prior distribution for $\theta$, which may potentially include hyperparameters for the likelihood (e.g., over-dispersion in a beta-binomial) in addition to those on parametrising the mean and covariance functions.  In hierarchical notation, supposing for clarity an iid Normal likelihood with variance, $\sigma_e^2$: 
\begin{eqnarray} \label{eq:eq1} 
\theta & \sim & \pi(\theta)\nonumber\\ 
f(s,t)|\mathbf{X}_{s,t},\theta & \sim & GP(\mu_{\theta},k_{\theta})\nonumber \\ 
y|f(s,t),\mathbf{X}_{s,t},\theta & \sim & N(f(s,t),\mathds{1}\sigma_e^2) 
\end{eqnarray} 
Following Bayes theorem the posterior distribution resulting from this hierarchy becomes: 
\begin{equation}\label{eq:posterior} 
\mathbb{P}(\theta,f(s,t)|y) = \frac{\mathbb{P}(y|f(s,t),\theta) \mathbb{P}(f(s,t)|\theta)\mathbb{P}(\theta)}{\int \int \mathbb{P}(y|f(s,t),\theta) \{d\mathbb{P}(f(s,t)|\theta)\} \{d\mathbb{P}(\theta)\}},
\end{equation} 
where the denominator in Equation \ref{eq:posterior} is the marginal likelihood, $\mathbb{P}(y)$.

Given the hierarchical structure in Equation \ref{eq:eq1} and the conditional properties of Gaussian distributions, the conditional predictive distribution for the mean of observations, $z$ $[=\mathbf{z}_{s^\prime,t^\prime}]$, at location-time pairs, $(s^\prime,t^\prime)[j]$, for a given $\theta$ is also Gaussian with form: 
\begin{eqnarray} 
z|y,\theta &\sim& N(\mu^\ast,\Sigma^\ast) \\ 
\label{eq:eq2a} \mu^\ast &=&  \mu_{(s^\prime,t^\prime)|\theta}+\Sigma_{(s^\prime,t^\prime),(s,t)|\theta}\Sigma_{y|(s,t),\theta}^{-1}\left(y-\mu_{(s,t)|\theta}\right)\\ 
 \label{eq:eq2b} \Sigma^\ast &=&  \Sigma_{(s^\prime,t^\prime)|\theta} - \Sigma_{(s^\prime,t^\prime),(s,t)|\theta}\Sigma_{y|(s,t),\theta}^{-1}\Sigma_{(s,t),(s^\prime,t^\prime)|\theta}
\end{eqnarray} 
where $\Sigma_{y|(s,t),\theta}= \left ( \Sigma_{\theta} +\mathds{1}\sigma_e^2\right)$. For specific details on the parametrisation of $Sigma$ see Appendix

When examining the conditional expectation in equation \ref{eq:eq2b} and splitting the summation into terms $\mu_{(s^\prime,t^\prime)|\theta}$ and $\Sigma_{(s^\prime,t^\prime),(s,t)|\theta}\Sigma_{y|(s,t),\theta}^{-1}\left(y-\mu_{(s,t)|\theta}\right)$,  it is clear that the first specifies a global underlying mean while the second augments the residuals from that mean by the covariance function. Clearly if the mean function fits the data perfectly the covariance in the second term of the expectation would drop out and conversely if the mean function is zero, then only the covariance function would model the data. This expectation therefore represents a balance between the underlying trend and the residual correlated noise.

In most applications of Gaussian process regression a linear mean function ($\mu_{\theta}=\mathbf{X}_{s,t}\mathbf{\beta}$) is used, where $\mathbf{\beta}$ is a vector of $m$ coefficients. However, when a rich suite of covariates is available this linear mean may be sub-optimal, limiting the generalisation accuracy of the model. To improve on the linear mean, covariate basis terms can be expanded to include parametric nonlinear transforms and interactions, but finding the optimal set of bases is computationally demanding and often leaves the researcher open to data snooping\cite{Abu-Mostafa2012LearningData}. In this paper we propose using an alternative two stage statistical procedure to first obtain a set of candidate non-linear mean functions using multiple different algorithmic methods fit without reference to the assumed spatial covariance structure and then include those means in the Gaussian process via stacked generalisation.

\subsection*{Stacked generalisation}\label{stackedgen}

Stacked generalisation \cite{Wolpert1992StackedGeneralization}, also called stacked regression \cite{BreimanStackedRegressions}, is a general ensemble approach to combining different models. In brief stacked generalisers combine different models together to produce a meta-model with equal or better predictive performance than the constituent parts \cite{vanderLaan2007SuperLearner}. In the context of malaria mapping our goal is to fuse multiple algorithmic methods with Gaussian process regression to both fully exploit the information contained in the covariates and model spatio-temporal correlations.

To present stacked generalisation we begin by introducing standard ensemble methods and show that stacked generalisation is simply a special case of this powerful technique. To simplify notation we suppress the spatio-temporal index and dependence on $\theta$. Consider $\mathcal{L}$ models, with outputs $\tilde{y}_i(x),i={1,...,\mathcal{L}}$. The choice of these models is described in appendix \ref{genmethods}. We denote the true target function as $f(x)$ and can therefore write the regression equation as $y_i(x)=f(x) + \epsilon_i(x)$. The average sum-of-squares error for model $i$ is defined as $E_i=\mathbb{E}[(\tilde{y}_i(x)-f(x))^2]$. Our goal is to estimate an ensemble model across all $\mathcal{L}$ models, denoted as $M(\tilde{y}_1,..,\tilde{y}_{\mathcal{L}})$. The simplest choice for $C$ is an average across all models $M(\tilde{y}_1,..,\tilde{y}_{\mathcal{L}})=\tilde{y}_\mathrm{avg}(x)=\frac{1}{\mathcal{L}}\sum_{i=1}^{\mathcal{L}} \tilde{y}_i(x)$. However this average assumes that the error of all models are the same, and that all models perform equally well. The assumption of equal performance may hold when using variants of a single model (i.e. bagging) but is unsuitable when very different models are used. Therefore a simple extension would be to use a weighted mean across models $M(\tilde{y}_1,..,\tilde{y}_{\mathcal{L}})=\tilde{y}_\mathrm{wavg}(x)=\sum_{i=1}^{\mathcal{L}} \beta_i \tilde{y}_i(x)$ subject to constraints $\beta>0 \; \forall i,\sum_{i=1}^{\mathcal{L}}\beta=1$ (convex combinations). These constraints prevent extreme predictions in well predicting models and impose the sensible inequality $\tilde{y}_\mathrm{min}(x)\leq \tilde{y}_\mathrm{wavg}(x)\leq \tilde{y}_\mathrm{max}$ \cite{BreimanStackedRegressions}. The optimal $\beta$s can be found by quadratic programming or by Bayesian linear regression with a Dirichlet/categorical prior on the coefficients. One particularly interesting result of combining models using this constrained weighted mean is the resulting decomposition of error into two terms \cite{Krogh95neuralnetwork} 
\begin{equation} 
\label{eq:biasvariance} 
\begin{split} 
\mathbb{E}[(\tilde{y}_\mathrm{wavg}(x)-f(x))^2] &=\sum_{i=1}^n \beta_i \mathbb{E}[(\tilde{y}_i(x)-f(x))^2] \\ & - \sum_{i=1}^n \beta_i \mathbb{E}[(\tilde{y}_i(x)-\tilde{y}_\mathrm{wavg}(x))^2] 
\end{split} 
\end{equation} Equation \ref{eq:biasvariance} is a reformulating of the standard bias-variance decomposition\cite{Geman1992NeuralDilemma} where first term describes the average error of all models, and the second (termed the ambiguity) is the spread of each member of the ensemble around the weighted mean, and measures the disagreement among models. Equation \ref{eq:biasvariance} shows that combining multiple models with low error but with large disagreements produces a lower overall error. It should be noted that equation \ref{eq:biasvariance} makes the assumption that $y(x)=f(x)$. 

Combination of models in an ensemble as described above can potentially lead to reductions in errors. However the ensemble models introduced so far are based only on training data and therefore neglect the issue of model complexity and tell us nothing about the ability to generalise to new data. To state this differently, the constrained weighted mean model will always allocate the highest weight to the model that most over fits the data. The standard method of addressing this issue is to use cross validation as a measure of the generalisation error and select the best performing of the $\mathcal{L}$ models. Stacked generalisation provides a technique to combine the power of ensembles described above, but also produces models that can generalise well to new data. The principle idea behind stacked generalisation is to train $\mathcal{L}$ models (termed level 0 generalisers) and generalise their combined behaviour via a second model (termed the level 1 generaliser). Practically this is done by specifying a $K$-fold cross validation set, training all $\mathcal{L}$ level 0  models on these sets and using the cross validation predictions to train a level 1 generaliser. This calibrates the level 1 model based on the generalisation ability of the level 0 models. After this level 1 calibration, all level 0 models are refitted using the full data set and these predictions are used in the level 1 model without refitting (This procedure is more fully described in algorithm \ref{alg:stgen} and the schematic design shown in Appendix figure \ref{fig:fig3}). The combination of ensemble modelling with the ability to generalise well has made stacking one of the best methods to achieve state-of-the art predictive accuracy \cite{Puurula2014KaggleSolution,BreimanStackedRegressions, vanderLaan2007SuperLearner}.

Defining the most appropriate level 1 generaliser based on a rigorous optimality criteria is still an open problem, with most applications using the constrained weighted mean specified above \cite{BreimanStackedRegressions,vanderLaan2007SuperLearner}. Using the weighted average approach can be seen as a general case of cross validation, where standard cross validation would select a single model by specifying a single $\beta_i$ as 1 and all other $\beta_i$s as zero. Additionally, it has been shown that using the constrained weighted mean method will perform asymptotically as well as the best possible choice among the family of weight combinations \cite{vanderLaan2007SuperLearner}.

Here we suggest using Gaussian process regression as the level 1 generaliser. Revisiting equation \ref{eq:eq2a} we can replace $\mu_{(s^\prime,t^\prime)|\theta}$ with a linear stacked function $\mu_{(s^\prime,t^\prime)|\theta}=\sum_{i=1}^{\mathcal{L}} \beta_i \tilde{y}_{i}(s^\prime,t^\prime)$ across $\mathcal{L}$ level 0 generalisers, where the subscript denotes predictions from the $i$th Level 0 generaliser (see algorithm \ref{alg:stgen}. We also impose inequality constraints on $\beta_i$ such that $\beta>0 \; \forall i,\sum_{i=1}^{\mathcal{L}}\beta=1$. This constraint allows the $\beta$s to approximately sum to one and helps computational tractability. It should be noted that empirical analysis suggests that simply imposing $\beta>0$ is normally sufficient \cite{BreimanStackedRegressions}.

The intuition in this extended approach is that the stacked mean of the Gaussian process uses multiple different methods to exploit as much predictive capacity from the covariates and then leaves the spatio-temporal residuals to be captured through the Gaussian process covariance function. In Appendix \ref{proof} we show that this approach yields all the benefits of using constrained weighted mean (equation \ref{eq:biasvariance}) but allows for a further reduction in overall error from the covariance function of the Gaussian process.

\subsection*{Data, Covariates and Experimental Design} \label{covs} The hierarchical structure most commonly used in infectious disease mapping is that shown in Equation  \ref{eq:eq1}. In malaria studies our response data are often discrete random variables representing the number of individuals testing positive for the \emph{Plasmodium falciparum}  (\emph{Pf}) malaria parasite, $N^{+}$, out of the total number tested, $N$, at a given location.  If the response is aggregated (e.g.) from the individual household level to a cluster or enumeration area level usually the centroid of the component sites is used as the spatial position datum.  The ratio of $N^{+}$ to $N$ is defined as the parasite rate or prevalence. To aid computation, the response data transformed to the empirical logit of observed prevalence \cite{Bhatt2015The2015.,Diggle2007Model-basedGeostatistics}, which for $N \gtrsim 20$ can be well approximated by a Gaussian likelihood. Pre-modelling standardisation of the available prevalence data for age and diagnostic type has also been performed on the sample used here, as described in depth in \cite{Bhatt2015The2015.,Gething2011A2010}. Our analysis is performed over Sub-Saharan Africa with the study area and dataset partitioned into 4 epidemiologically-distinct regions \cite{Gething2011A2010}---Eastern Africa, Western Africa, North Eastern Africa, and Southern Africa---each of which was modelled separately (see Figure \ref{fig:fig1}.

\begin{figure}[] \centering  \includegraphics[width=0.7\textwidth]{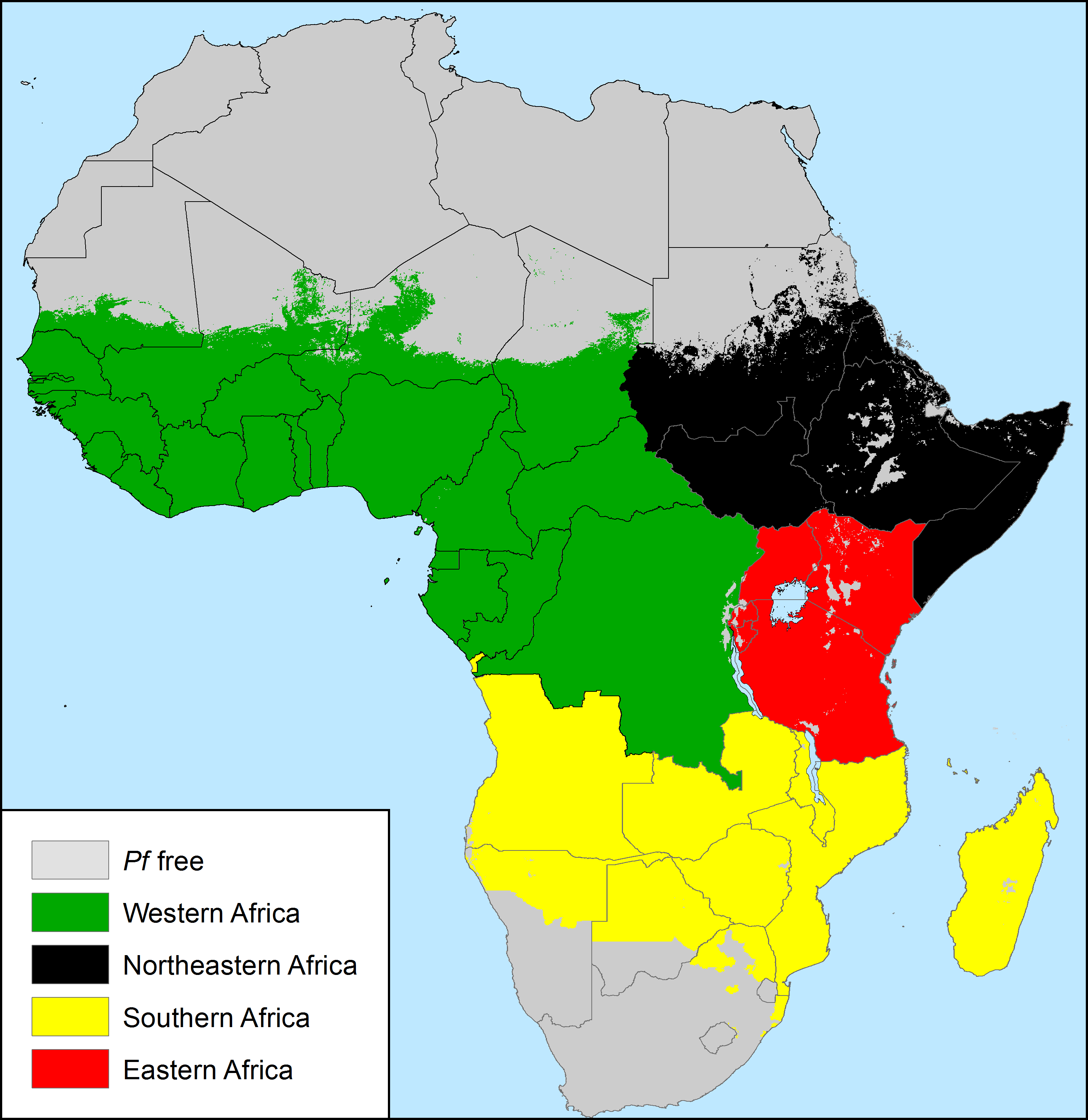} \caption{Study area of stable malaria transmission in Sub-Saharan Africa. Our analysis was performed on 4 zones - Western Africa, North eastern Africa, Eastern Africa and Southern Africa} \label{fig:fig1} \end{figure}

The covariates (i.e., independent variables) used in this research consist of raster layers spanning the entire continent at a 2.5 arc-minute (5 km x 5 km) spatial resolution. The majority of these raster covariates were derived from high temporal resolution satellite images that were first gap filled \cite{Weiss2014AnTime-series.} to eliminate missing data (resulting primarily from persistent cloud cover over equatorial forests) and then aggregated to create a dynamic (i.e. temporally varying) dataset for every month throughout the study period (2000-2015). The list of covariates is presented in Table \ref{table1} and detailed information on individual covariates can be found here \cite{Weiss2015Re-examiningApproach.,Weiss2014AirPrediction., Weiss2014AnTime-series.}. The set of monthly dynamic covariates was further expanded to include lagged versions of the covariate at 2 month, 4 month and 6 month lags. The main objective of this study was to judge the predictive performance of the various generalisation methods and therefore no variable selection or thinning of the covariate set was performed. It should be noted however that many of the level 0 generalisers performed variable selection automatically (e.g. elastic net regression).

The level 0 generalisers used were gradient boosted trees \cite{Friedman2001GreedyMachine,Friedman2002StochasticBoosting}, random forests\cite{breiman2001random}, elastic net regularised regression\cite{Zou2005RegularizationNet}, generalised additive splines \cite{Wood2006GeneralizedR,Hastie2009TheLearning} and multivariate adaptive regression splines\cite{Friedman1991MultivariateSplines}. The level 1 generalisers used were stacking using a constrained weighted mean and stacking using Gaussian process regression. We also fitted a standard Gaussian process for benchmark comparisons with the level 0 and 1 generalisers. Stacked fitting was performed following algorithm \ref{alg:stgen}. Full analysis and $K$-fold Cross validation was performed 5 times and then averaged to reduce any bias from the choices of cross validation set. The averaged cross validation results were used to estimate the generalisation error by calculating the mean squared error (MSE), mean absolute error (MAE) and the correlation.

\pagebreak

\section*{Tables and Algorithms} 

\begin{algorithm} \caption{Stacked Generalisation Algorithm: The algorithm proceeds as follows. In line 2 to 4 the covariates,response and number of cross validation folds is defined. Lines 6 to 9 fits all level 0 generalisers to the full data set. Lines 10 to 16 fits all level 0 generalisers to cross validation data sets. Line 17 to 18 fits a level 1 generaliser to the cross validation predictions and Line 19 returns the final output by using the level 1 generaliser to predict on the full predictions}\label{alg:stgen} \begin{algorithmic}[1] \Procedure{Stack}{}\Comment{covariate and response input} \State \textbf{Input} $X$ as a $n\times m$ design matrix \State \textbf{Input} $y$ as a $n$ vector of responses \State \textbf{Input} $v$ cross validation folds \State \textbf{choose} $l$, $ \mathcal{L}(y,X)$ models\Comment{level zero generalisers} \State \textbf{define} $n\times l$ matrix $P$ \Comment{matrix of predictions} \For{$i\gets 1, l$} \State \textbf{fit} $\mathcal{L}_i(y,X)$ \State \textbf{predict} $P_{\cdot,i} = \mathcal{L}_i(y,X)$ \EndFor

\State \textbf{split} $X,y$ into $\{g_1,...g_v\}$  groups $\{X_{g_1},..,X_{g_v}\}$ and $\{y_{g_1},..,y_{g_v}\}$ \Comment {training set} \State \textbf{add} remaining samples to $\{X_{/g_1},..,X_{/g_v}\}$ and $\{y_{/g_1},..,y_{/g_v}\}$ \Comment {testing set} \State \textbf{define} $n\times l$ matrix $H$ \Comment{matrix cross validation of predictions} \For{$i\gets 1, l$} \For{$j\gets 1, v$} \State \textbf{fit} $\mathcal{L}_i(y_{g_j},X_{g_j})$ \State \textbf{predict} $H_{/g_j,i} = \mathcal{L}_i(y_{/g_j},X_{/g_j})$ \EndFor \EndFor \State \textbf{choose} $\mathcal{L}^*(y,H)$ model \Comment{level one generaliser} \State \textbf{fit} $\mathcal{L}^*(y,H)$  \\ \Return {$\mathcal{L}^*(y,P)$} \Comment{final prediction output} \EndProcedure \end{algorithmic} \end{algorithm}

\begin{table}[H] \caption{List of Environmental, Socio-demographic and Land type covariates used.} \label{table1} \begin{tabular}{|>{\centering}p{4cm}|>{\centering}p{4cm}|>{\centering}p{ 2cm}|>{\centering}p{4cm}|} \hline \multicolumn{1}{|>{\centering}p{4cm}|}{Variable Class} & \multicolumn{3}{c|}{Variable(s)SourceType}\tabularnewline \hline \hline Temperature & Land Surface Temperature (day, night, and diurnal-flux) & MODIS Product & Dynamic Monthly\tabularnewline \hline Temperature Suitability & Temperature Suitability for Plasmodium falciparum & Modeled Product  & Dynamic Monthly\tabularnewline \hline Precipitation & Mean Annual Precipitation & WorldClim  & Synoptic\tabularnewline \hline Vegetation Vigor & Enhanced Vegetation Index & MODIS Derivative & Dynamic Monthly\tabularnewline \hline Surface Wetness & Tasseled Cap Wetness & MODIS Derivative & Dynamic Monthly\tabularnewline \hline Surface Brightness  & Tasseled Cap Brightness  & MODIS Derivative & Dynamic Monthly\tabularnewline \hline IGBP Landcover & Fractional Landcover  & MODIS Product & Dynamic Annual\tabularnewline \hline IGBP Landcover Pattern & Landcover Patterns  & MODIS Derivative & Dynamic Annual\tabularnewline \hline Terrain Steepness & SRTM Derivatives & MODIS Product & Static\tabularnewline \hline Flow \& Topographic Wetness & Topographically Redistributed Water & SRTM Derivatives & Static\tabularnewline \hline Elevation & Digital Elevation Model & SRTM & Static\tabularnewline \hline Human Population & AfriPop  & Modeled Products & Dynamic Annual\tabularnewline \hline Infrastructural Development & Accessibility to Urban Centers and Nighttime Lights & Modeled Product and VIIRS  & Static\tabularnewline \hline Moisture Metrics & Aridity and Potential Evapotranspiration & Modeled Products & Synoptic\tabularnewline \hline \end{tabular} \end{table}

\nolinenumbers

\bibliography{Mendeley.bib}
\pagebreak

\subsection*{Funding Statement}
SB is supported by the MRC outbreak centre.
DLS was supported by the Bill and Melinda Gates Foundation 
$[OPP1110495]$, National Institutes of Health/National Institute
of Allergy and Infectious Diseases $[U19AI089674]$, and the
Research and Policy for Infectious Disease Dynamics (RAPIDD)
program of the Science and Technology Directorate, Department of
Homeland Security, and the Fogarty International Center, National
Institutes of Health. PWG is a Career Development Fellow $[K00669X]$ jointly funded by the UK Medical Research Council (MRC) and the UK Department for International Development (DFID) under the MRC/DFID Concordat agreement, also part of the EDCTP2 programme supported by the European Union, and receives support from the Bill and Melinda Gates Foundation $[OPP1068048, OPP1106023]$. These grants also support DJW, and EC.

\pagebreak

\section{S1 - Appendix}

\subsection{S1 - Generalisation methods}\label{genmethods}

We choose 5 level 0 generalisers to feed into the level 1 generaliser. We chose these 5 due to (a) ease of implementation through existing software packages, (b) the differences in their approaches and (c) a proven track record in predictive accuracy.

\subsubsection{Gradient boosted trees and Random forests}

Both Gradient boosted trees and random forests produce ensembles of regression trees \cite{Breiman1984ClassificationTrees}. Regression trees partition the space of all joint covariate variable values into disjoin regions $R_j$ ($j=\{1,2,...,J\}$) which are represented as terminal nodes in a decision tree. A constant $\gamma_j$ is assigned to each region such that the predictive rule is $x\in R_j \rightarrow f(x) = \gamma_j$ \cite{Hastie2009TheLearning}. A tree is therefore formally expressed as $T(x,\theta)=\sum_{j=1}^J \gamma_j \mathbb{I}(x\in R_j)$.

Gradient boosted trees model the target function by a sum of such trees induced by forward stagewise regression \begin{eqnarray} \label{eq:gbm} f_M(x)&=&\sum_{m=1}^M T(x,\theta_m)\\ \hat{\theta}_m&=&\argmin_{\theta_m}\sum_{i=1}^N L(y_i,f_{m-1}(x_i) + T(x_i,\theta_m)) \end{eqnarray} Solving equation \ref{eq:gbm} is done by functional gradient decent with regularisation performed via shrinkage and cross validation \cite{Friedman2001GreedyMachine}. In our implementations we also stochastically sampled covariates and samples in each stagewise step \cite{Friedman2002StochasticBoosting}.

Random forests combine trees by bagging \cite{Hastie2009TheLearning} where bootstrap samples ($B$) of covariates and data are used to create an average ensemble of trees:

\begin{equation} \label{rf} f_M(x) = \frac{1}{B}\sum_{b=1}^B T(x,\theta_b) \end{equation}

The optimal number of bootstraped trees $B$ is found by cross validation\cite{BreimanRandomForests}.

We used the H2O package in R to fit both the gradient boosted models and random forest models and meta-parameters such as tree depth and samples per node were evaluated using a coarse grid search.

\subsubsection{Elastic net regularised regression}

Elastic net regularised regression \cite{Zou2005RegularizationNet} is a penalised linear regression where the coefficients of the regression are found by 
\begin{eqnarray}
\label{elastic} f_M(x) &=& X^T\hat{\beta}\\ \hat{\beta}&=&argmin_\beta||y-f_M(x)||^2 + \lambda_2||\beta||^2 + \lambda_1||\beta||_1 
\end{eqnarray} Where the subscripts on the penalty terms represent the $\ell_1$ (i.e. lasso) and $\ell_2$ (i.e. ridge) norms. These norms induce shrinkage in the coefficients of the linear model allowing for better generalisation.

Equation \ref{elastic} is fitted by numerically maximising the likelihood and optimal parameters for $\lambda_1,\lambda_2$ are computed by cross validation over the full regularisation path. Fitting was done using the H2O package in R.

\subsubsection{Generalised additive splines}

Generalised additive splines extend standard linear models by allowing the linear terms to be modelled as nonlinear flexible splines \cite{Hastie1990GeneralizedModels}.
\begin{eqnarray}
\label{gam} f_M(x) &=& \beta_0 + f_1(x_1)+,..,+f_m(x_m)\\ &s.t.& argmin_\beta||y-f_M(x)||^2 + \lambda \int (f^{''}_M(x))^2 dx 
\end{eqnarray}
Where $f^{''}$ denotes the second derivative and penalises non smooth functions (that can potentially overfit). Fitting was done using generalised cross validation with smoothness selection done via restricted maximum likelihood. Fitting was done using the mgcv package in R.

\subsubsection*{Multivariate adaptive regression splines} Multivariate adaptive regression splines \cite{Friedman1991MultivariateSplines} build a model of the form \begin{equation} f_M(x)=\sum_{i=m}^M \beta_i V_i(x) \end{equation}

Where $\beta_i$s are coefficients and $V_i$s are basis functions that can either be constant, be a hinge function of the form $max(0,x-const),max(0,const-x)$, or the product of multiple hinge functions. Fitting is done using a two stage approach with a forward pass adding functions and a backward pruning pass via generalised cross validation. Fitting was done using the earth package in R.

\subsection{S1 - Details of the Gaussian Process}\label{GMRFmethod}

For the Gaussian process spatial component the covariance function is chosen to be Mat\'{e}rn of smoothness $\nu = 1$: \begin{equation} k_{\theta}(s_i,s_j)=\kappa/\tau\|s_i-s_j\|\mathcal{K}_1^{(2)} (\kappa\|s_i-s_j\|), \end{equation} with $\kappa=\sqrt{2}/\rho$ an inverse range parameter (for range, $\rho$), $\tau$ a precision (i.e., inverse variance) parameter, and $\mathcal{K}_1^{(2)}$ the modified Bessel function of the second kind and order 1.  Typically, $\theta$ is defined with elements $\{\log \kappa,\ \log \tau\}$ to ensure positivity via the exponential transformation.  For computational efficiency and scalability we follow the stochastic partial differential equation (SPDE) approach\cite{Lindgren2011AnApproach} to approximate the continuous Gaussian process in Equation \ref{eq:eq1} with a discrete Gauss-Markov random field (GRMF) of sparse precision matrix, $Q_\theta$ $[=\Sigma_\theta^{-1}]$, allowing for fast computational matrix inversion \cite{Rue2009ApproximateApproximations}. To find the appropriate GMRF structure, the Mat\'{e}rn \cite{Rasmussen2006GaussianLearning} spatial component is parametrised as the finite element solution to the SPDE, $(k^{2}-\triangle)(\tau f(s))=\mathcal{W}(s)$ defined on a spherical manifold, $\mathbb{S^\mathrm{2}}\in\mathbb{R^{\mathrm{3}}}$, where $\triangle=\frac{\partial}{\partial{s_{(1)}^2}}+\frac{\partial}{\partial{s_{(2)}^2}} +\frac{\partial}{\partial{s_{(3)}^2}}$ is the Laplacian (for Cartesian coordinates $s_{(1)},s_{(2)},s_{(3)})$ and $\mathcal{W}(s)$ is a spatial white noise process\cite{Lindgren2011AnApproach}.

To extend this spatial process to a spatio-temporal process, temporal innovations are modelled by first order autoregressive dynamics: 
\begin{equation} 
f(s_i,t)=\phi f(s_i,t-1) + \omega(s_i,t),
\end{equation} 
where $|\phi|<1$ is the autoregressive coefficient and $\omega(s_i,t)$ is iid Normal.  Practically the spatio-temporal process is achieved through a Kronecker product of the (sparse) spatial SPDE precision matrix and the (sparse) autoregressive temporal precision matrix.

In this paper for readability and consistency with standard notation equations \ref{eq:eq2a} and \ref{eq:eq2b} are parameterised through the covariance matrix. As specified above in our GMRF implementation we instead specified the precision (inverse covariance) matrix. Using the precision matrix the conditional predictive distribution takes the form
\begin{eqnarray} 
z|y,\theta &\sim& N(\mu^\ast,Q^{-1\ast}) \\ 
\label{eq:eq2GMRF} \mu^\ast &=&  \mu_{(s^\prime,t^\prime)|\theta}+
Q_{(s^\prime,t^\prime),(s^\circ ,t^\circ)|\theta}^{-1}
A^{T}
Q_{y|(s^\circ ,t^\circ),\theta}
\left(y-A\mu_{(s^\circ ,t^\circ)|\theta}\right)\\ 
Q^{-1\ast} &=& Q_{(s^\prime,t^\prime),(s^\circ ,t^\circ)|\theta}^{-1}
\end{eqnarray} 
Where $Q_{(s^\prime,t^\prime),(s^\circ ,t^\circ)|\theta}=Q_{(s^\prime,t^\prime)|\theta}+A^{T}Q_{y|(s^\circ ,t^\circ),\theta}A$. In Equation \ref{eq:eq2GMRF} $A$ is introduced as a sparse observation matrix that maps the finite dimensional GMRF at locations $(s^\circ ,t^\circ)$ to functional evaluations at any spatio-temporal locations e.g prediction locations $(s^\prime,t^\prime)$ or data locations $(s,t)$, provided these locations are within a local domain.

\subsection{S1 - Bias variance derivation for a Gaussian process stacked generaliser}\label{proof} 

\begin{theorem}
Consider a function $f$ from $\mathbb{R}^N$ to $\mathbb{R}$ for which a sample $D=\{x_i,y_i\}$ exists, where $y_i = f(x_i)$ and $i=\{1,...,n\}$.  Fit $\mathcal{L}$ level 0 generalisers, $M_1(x),.., M_{\mathcal{L}}(x)$, trained on data $D$. Next define two ensembles of the $\mathcal{L}$ models, the first using a weighted mean, $\bar{M}_{cwm}(x)=\sum_{i=1}^{\mathcal{L}} \beta_i M_i(x)$, and the second as the mean of a Gaussian process, $\bar{M}_{gp}(x)=\sum_{i=1}^{\mathcal{L}} \beta_i M_i(x) + \Sigma_2\Sigma_1^{-1} \left(f(x) - \sum_{i=1}^{\mathcal{L}} \beta_i M_i(x) \right)$, with $\beta$ subject to convex combinations for both models. If the squared error is taken for both ensembles i.e. $e_{cwm}(x)= (f(x) - \bar{M}_{cwm}(x))^2$ and $e_{gp}(x)= (f(x) - \bar{M}_{gp}(x))^2$, then from the contribution of the covariance $( \mathbb{I} -\Sigma_{2}\Sigma_{1})e_{gp}(x) \leq e_{cwm}(x) \; \forall x$
\end{theorem}

\begin{proof}
Consider a function $f$ from $\mathbb{R}^N$ to $\mathbb{R}$ for which a sample $D=\{x_i,y_i\}$ exists, where $y_i = f(x_i)$ and $i=\{1,...,n\}$. fit $\mathcal{L}$ level 0 generalisers, $M_1(x),..,M_L(x)$ (see algorithm \ref{alg:stgen}). 

Consider two ensembles, $\bar{M}$, of the $M_1(x),..,M_{\mathcal{L}}(x)$ models
\begin{eqnarray}
 \label{eq:model1}\bar{M}_{cwm}(x)&=&\sum_{i=1}^{\mathcal{L}} \beta_i M_i(x) \\
 \label{eq:model2}\bar{M}_{gp}(x)&=&\sum_{i=1}^{\mathcal{L}} \beta_i M_i(x) + \Sigma_2\Sigma_1^{-1} \left(f(x) - \sum_{i=1}^L \beta_i M_i(x) \right)
\end{eqnarray}

With convex combination constraints $\beta_i\geq 0\forall i$ and $\sum_{i=1}^L \beta_i = 1$.

Model 1 (equation \ref{eq:model1}), referred too here as a constrained weighted mean, is the predominant ensemble approach taken previously \cite{BreimanStackedRegressions,vanderLaan2007SuperLearner,Krogh95neuralnetwork,Sill2009Feature-WeightedStacking}. Model 2 (equation \ref{eq:model2}), is the conditional expectation of ensembling via Gaussian process regression see (equation \ref{eq:eq2a} and \cite{Rasmussen2006GaussianLearning}). To simplify notation here $\Sigma_2 = \Sigma_{(s^\prime,t^\prime),(s,t)|\theta}$ and $\Sigma_1 = \Sigma_{y|(s,t),\theta}$.

Define the squared error between the target function and each individual level zero generaliser as $\epsilon_i(x) = (f(x) - M_i(x))^2$. Define the squared error between the target function and the ensemble as $e(x)= (f(x) - \bar{M}(x))^2$. Following from \cite{Krogh95neuralnetwork} define the ambiguity of a given model as $a_i(x) = (\bar{M}(x)-M_i(x))^2$.

As derived in\cite{Krogh95neuralnetwork} the ensemble ambiguity from using a constrained weighted mean ensemble (model 1 above) subject to convex combinations is defined as the weighted sum of the individual ambiguities. 
\begin{eqnarray} 
\bar{a}(x) &=& \sum_{i=1}^{\mathcal{L}} \beta_i a_i(x) = \sum_{i=1}^{\mathcal{L}} \beta_i(\bar{M}(x)-M_i(x))^2  \\
\label{eq:cwmsensemble} \bar{a}(x) &=& \bar{\epsilon}(x) - e(x)
\end{eqnarray} 
where $ \bar{\epsilon}(x) = \sum_{i=1}^{\mathcal{L}} \beta_i \epsilon_i(x)$. Therefore the ensemble squared error when using a constrained weighted mean is $e_{cwm}(x) = \bar{\epsilon}(x) - \bar{a}(x)$.

Using the Gaussian process (model 2 above) the ensemble ambiguity is defined as.
\begin{eqnarray} 
\label{eq:gpensem1}\bar{a}(x) &=& \sum_{i=1}^{\mathcal{L}} \beta_i (\bar{M}(x)-M_i(x))^2 \nonumber \\ & +& \Sigma_{2^*}\Sigma_{1^*}^{-1} \Bigg( \sum_{i=1}^{\mathcal{L}} \beta_i  \big( f(x) - M_i(x)\big)^2 - \sum_{i=1}^{\mathcal{L}} \beta_i \big(\bar{M}(x) - M_i(x)\big)^2   \Bigg)\\ 
\label{eq:gpensem2}\bar{a}(x) &=& \sum_{i=1}^{\mathcal{L}} \beta_i a_i(x) + \Sigma_{2^*}\Sigma_{1^*}^{-1} \Bigg( \sum_{i=1}^{\mathcal{L}} \beta_i e_i(x) -  \sum_{i=1}^{\mathcal{L}} \beta_i a_i(x)\Bigg) 
\end{eqnarray} 
In equations \ref{eq:gpensem1} and \ref{eq:gpensem2} above  the covariance matrices operate on the ambiguities, and as such are suffixed with asterixes to distinguish those from equation \ref{eq:model2}. Additionally $\beta$ in equations \ref{eq:gpensem1} and \ref{eq:gpensem2} are also subject to convex combination constraints. Substituting \ref{eq:cwmsensemble} into equation \ref{eq:gpensem2} yields:
\begin{eqnarray} 
\bar{a}(x) &=& \bar{\epsilon}(x) - e(x) + \Sigma_{2^*}\Sigma_{1^*}^{-1}\Big(\bar{\epsilon}(x) - \bar{\epsilon}(x) + e(x)\Big)\\ 
\bar{a}(x) &=& \bar{\epsilon}(x) - e(x) + \Sigma_{2^*}\Sigma_{1^*}^{-1} e(x)\\ 
\Big( \mathbb{I} - \Sigma_{2^*}\Sigma_{1^*}^{-1}\Big)e_{gp}(x) &=& \bar{\epsilon}(x) - \bar{a}(x)\label{error} 
\end{eqnarray} 
The right hand side of equation \ref{error} is identical to that derived using a constrained weighted mean in equation \ref{eq:cwmsensemble}, but the left hand side error term, $e_{gp}(x)$ is augmented by $(\mathbb{I} - \Sigma_{2^*}\Sigma_{1^*}^{-1})$. Clearly $\bar{a}(x)\leq\bar{\epsilon}(x) \; \forall x$, with $\bar{a}(x)=\bar{\epsilon}(x)$ only when the ensemble equals the true target function, $\bar{M}=f(x)$, and $e(x)=0$. It follows that the left hand side of equation \ref{error} $ \Big( \mathbb{I} -\Sigma_{2^*}\Sigma_{1^*}\Big)e(x)\geq 0 \; \forall x$. Therefore from the contribution of the precision or covariance stacking using a Gaussian process approach always has a lower error than stacking via a constrained weighted mean, with the error terms being equal when the contribution of the covariance is zero. That is $ ( \mathbb{I} -\Sigma_{2^*}\Sigma_{1^*})e_{gp}(x) \leq e_{cwm}(x) \; \forall x$
\end{proof}

\subsection{Alternative stacking designs}\label{alternative}

In section \ref{stackedgen} we presented the rational for stacking and introduced a basic design (design 1 in figure \ref{fig:fig3}) where multiple level 0 generalisers are stacked through a single level 1 generaliser. For this design we proposed a Gaussian process or constrained weighted mean as the level 1 generaliser. In figure \ref{fig:fig3} we suggest two alternative stacking designs. 

In design 2, multiple level 0 generalisers are fitted and then passed through  individual level 1 generalisers before being finally combined in a level 2 generaliser. An example of this scheme would be to fit multiple level 0 generalisers using different algorithmic methods (see Appendix \ref{genmethods}) and then feed each of these into a Gaussian process regression. This design allows for the Gaussian processes to learn individual covariance structures for each level 0 algorithmic method (as opposed to a joint structure as in design 1). These level 1 Gaussian processes can then be combined through a constrained weighted mean level 2 generaliser. 

In design 3, a single level 0 generaliser is used and fed into multiple level 1 generalisers before being combined via a level 2 generaliser. An example of this scheme would be to fit a single level 0 method, such as a linear mean or random forest, and then feed this single generaliser into multiple level 1 Gaussian processes. These multiple Gaussian processes can learn different aspects of the covariance structure such as long range, short range or seasonal interactions. These level 1 generalisers can then be combined, as in design 2, through a constrained weighted mean level 2 generaliser.  

\begin{figure}[] 
\centering 
\includegraphics[width=0.6\textwidth]{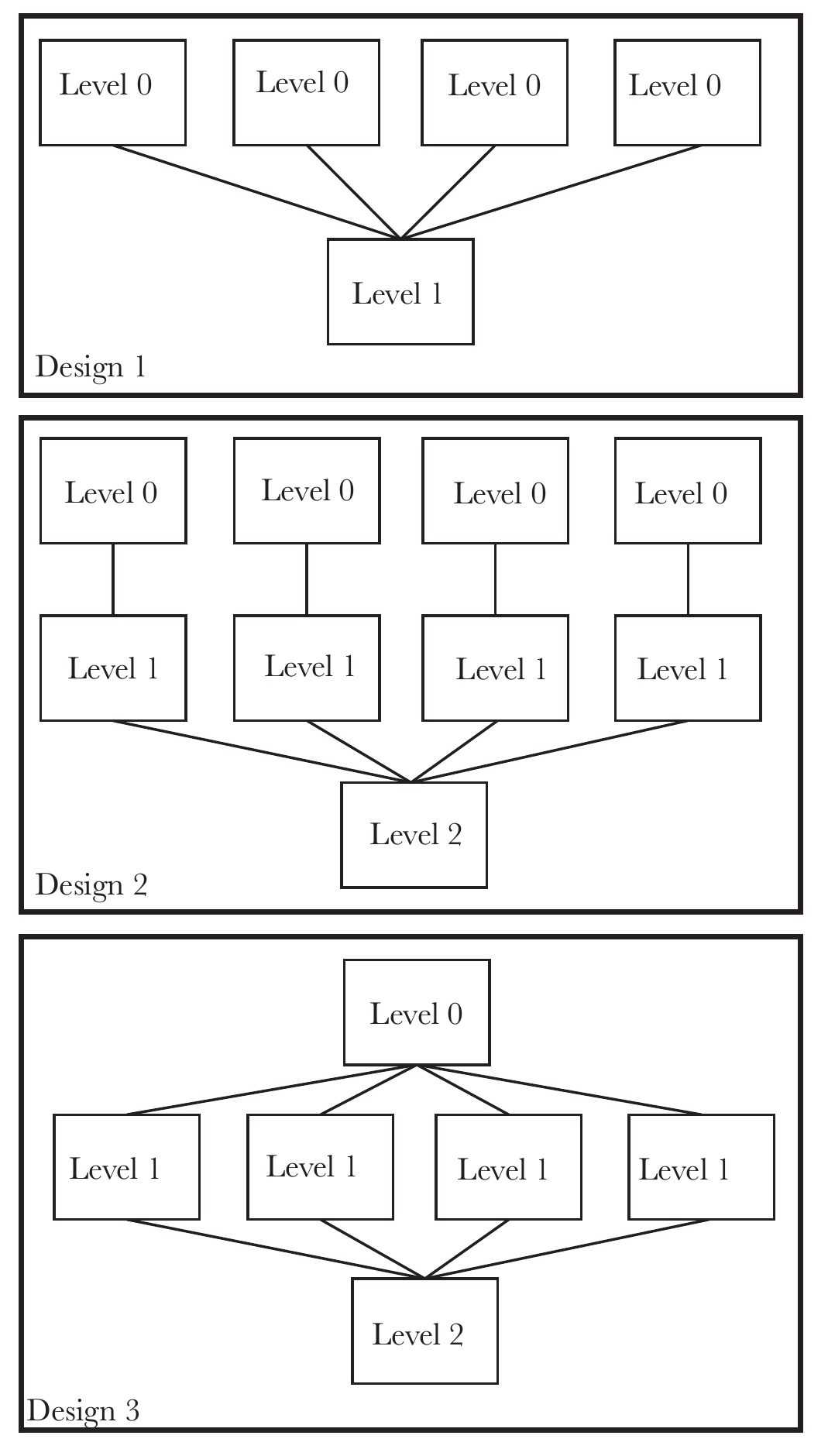}
\caption{Three suggested stacking designs. In this paper we have exclusively used design 1, where multiple level 0 generalisers are combined through a level 1 generaliser. An alternative (design 2) is to feed each level 0 generaliser into a unique level 1 generaliser and then combine them together through a level 2 generaliser. Another alternative (design 3) is to have a single level 0 generaliser feeding into multiple level 1 generalisers that are then combined through a level 2 generaliser} 
\label{fig:fig3} 
\end{figure}

\end{document}